\newif\ifconf\conftrue

\conffalse

\ifconf
\documentclass[10pt,conference,compsocconf]{IEEEfocs}
\IEEEoverridecommandlockouts
\else
\documentclass[11pt,letterpaper]{article}
\fi

\usepackage{amsthm}
\usepackage{times}
\usepackage[margin=1in,dvips]{geometry}

\usepackage{graphicx}
\usepackage{amssymb}
\usepackage{amsmath}
\usepackage{float}
\usepackage{url}
\usepackage{enumerate}

\usepackage[boxed,section]{algorithm}
\usepackage{algpseudocode}

\newtheorem{theorem}{Theorem}[section]
\newtheorem{lemma}[theorem]{Lemma}
\newtheorem{corollary}[theorem]{Corollary}
\newtheorem{definition}[theorem]{Definition}

\newtheorem{remark}[theorem]{Remark}
\newtheorem{fact}[theorem]{Fact}

\newcommand{\abs}[1]{\left|#1\right|}

\newcommand{\norm}[2]{\left \lVert#2\right \rVert_{#1}}

\newcommand{\poly}{{\mathrm{poly}}}

{\makeatletter
 \gdef\xxxmark{%
   \expandafter\ifx\csname @mpargs\endcsname\relax 
     \expandafter\ifx\csname @captype\endcsname\relax 
       \marginpar{xxx}
     \else
       xxx 
     \fi
   \else
     xxx 
   \fi}
 \gdef\xxx{\@ifnextchar[\xxx@lab\xxx@nolab}
 \long\gdef\xxx@lab[#1]#2{{\bf [\xxxmark #2 ---{\sc #1}]}}
 \long\gdef\xxx@nolab#1{{\bf [\xxxmark #1]}}
 \long\gdef\xxx@lab[#1]#2{}\long\gdef\xxx@nolab#1{}%
}

\DeclareMathOperator*{\argmax}{arg\,max}

\DeclareMathOperator{\E}{\mathbb{E}}
\newcommand\R{\mathbb{R}}
\newcommand\eps{\epsilon}
\DeclareMathOperator{\err}{Err^2}

\begin{document}

\title{On the Power of Adaptivity in Sparse Recovery}

\ifconf
\else
\author{Piotr Indyk \and Eric Price \and David P. Woodruff}
\fi

\ifconf
\else
\begin{titlepage}
\fi

\maketitle
\begin{abstract}
The goal of (stable) sparse recovery is to recover a $k$-sparse
approximation $x^*$ of a vector $x$ from linear measurements of
$x$. Specifically, the goal is to recover $x^*$ such that
$$\norm{p}{x-x^*} \le C \min_{k\text{-sparse } x'} \norm{q}{x-x'}$$
for some constant $C$ and norm parameters $p$ and $q$. It is known
that, for $p=q=1$ or $p=q=2$, this task can be accomplished using
$m=O(k \log (n/k))$ {\em non-adaptive}
measurements~\cite{CRT06:Stable-Signal} and that this bound is
tight~\cite{DIPW,FPRU,PW11}.

In this paper we show that if one is allowed to perform measurements
that are {\em adaptive} , then the number of measurements can be
considerably reduced. Specifically, for $C=1+\epsilon$ and $p=q=2$ we
show
\begin{itemize}
\item A scheme with $m=O(\frac{1}{\eps}k \log \log (n\eps/k))$
  measurements that uses $O(\log^* k \cdot \log \log (n\eps/k))$
  rounds.  This is a significant improvement over the best
    possible non-adaptive bound.
\item A scheme with $m=O(\frac{1}{\eps}k \log (k/\eps) + k \log
  (n/k))$ measurements that uses {\em two} rounds. This improves over
  the best possible non-adaptive bound.
\end{itemize}

To the best of our knowledge, these are the first results of this type. 

\ifconf
\else
As an independent application, we show how to solve the problem
of finding a duplicate in a data stream of $n$ items drawn
from $\{1, 2, \ldots, n-1\}$ using $O(\log n)$
bits of space and $O(\log \log n)$ passes, improving over the best
possible space complexity achievable using a single pass. 
\fi

\end{abstract}

\ifconf
\else
\thispagestyle{empty}
\end{titlepage}
\fi

\section{Introduction}

In recent years, a new ``linear'' approach for obtaining a succinct
approximate representation of $n$-dimensional vectors (or signals) has
been discovered.  For any signal $x$, the representation is equal to
$Ax$, where $A$ is an $m \times n$ matrix, or possibly a random
variable chosen from some distribution over such matrices.  The vector
$Ax$ is often referred to as the {\em measurement vector} or {\em
  linear sketch} of $x$.  Although $m$ is typically much smaller than
$n$, the sketch $Ax$ often contains plenty of useful information about
the signal $x$.

A particularly useful and well-studied problem is that of {\em stable
  sparse recovery}. We say that a vector $x'$ is $k$-sparse if it has at most $k$ non-zero
coordinates.  The sparse recovery problem is typically defined as follows: for
some norm parameters $p$ and $q$ and an approximation factor $C>0$,
given $Ax$, recover an ``approximation'' vector $x^*$ such that
\begin{equation}
\label{e:lplq}
\norm{p}{x-x^*} \le C \min_{k\text{-sparse } x'}  \norm{q}{x-x'}
\end{equation}
(this inequality is often referred to as {\em $\ell_p/\ell_q$
  guarantee}).  If the matrix $A$ is random, then
Equation~\eqref{e:lplq} should hold for each $x$ with some probability
(say, 2/3).  Sparse recovery has a tremendous number of applications
in areas such as compressive sensing of
signals~\cite{CRT06:Stable-Signal, Don06:Compressed-Sensing}, genetic
data acquisition and analysis~\cite{SAZ,BGSW} and data stream
algorithms\footnote{In streaming applications, a data stream is
  modeled as a sequence of linear operations on an (implicit) vector
  x. Example operations include increments or decrements of $x$'s
  coordinates. Since such operations can be directly performed on the
  linear sketch $Ax$, one can maintain the sketch using only $O(m)$
  words.} \cite{Muthu:survey2, I-SSS}; the latter includes
applications to network monitoring and data analysis. 

It is known~\cite{CRT06:Stable-Signal}
that there exist matrices $A$ and associated recovery algorithms that
produce approximations $x^*$ satisfying Equation~\eqref{e:lplq} with
$p=q=1$, constant approximation factor $C$, and sketch length
\begin{equation}
\label{e:m}
m=O(k \log (n/k))
\end{equation}
A similar bound, albeit using random matrices $A$, was later obtained
for $p=q=2$ \cite{GLPS} (building
on~\cite{CCF,CM03b,CM06:Combinatorial-Algorithms}).  Specifically, for
$C=1+\epsilon$, they provide a distribution over matrices $A$ with
\begin{equation}
\label{e:m2}
m=O(\frac{1}{\eps}k\log(n/k))
\end{equation}
rows, together with an associated recovery algorithm.

It is also known that the bound in~Equation~\eqref{e:m} is
asymptotically optimal for some constant $C$ and $p=q=1$,
see~\cite{DIPW} and~\cite{FPRU} (building on~\cite{GG,G,Kas}).  The
bound of \cite{DIPW} also extends to the randomized case and $p=q=2$.
For $C=1+\epsilon$, 
a lower bound of $m=\Omega(\frac{1}{\eps}k \log(n/k))$ 
was recently shown~\cite{PW11} for the randomized case and $p = q = 2$, 
improving upon the earlier work of \cite{DIPW} and 
showing the dependence on $\epsilon$ is optimal. 
The necessity of the ``extra'' logarithmic factor multiplying $k$ is
quite unfortunate: the sketch length determines the ``compression
rate'', and for large $n$ any logarithmic factor can worsen that rate
tenfold.

In this paper we show that this extra factor can be greatly reduced if
we allow the measurement process to be {\em adaptive}. In the adaptive
case, the measurements are chosen in rounds, and the choice of the
measurements in each round depends on the outcome of the measurements
in the previous rounds. The adaptive measurement model has received a
fair amount of attention in the recent
years~\cite{JXC,CHNR,HCN,HBCN,MSW, AWZ}, see also~\cite{KeCoM}.  In
particular~\cite{HBCN} showed that adaptivity helps reducing the
approximation error in the presence of random noise. However, no
asymptotic improvement to the number of measurements needed for sparse recovery 
(as in Equation~\eqref{e:lplq})
 was previously known.

\paragraph{Results}
In this paper we show that adaptivity can lead to very significant
improvements in the number of measurements over the bounds in
Equations \eqref{e:m} and \eqref{e:m2}. We consider randomized sparse
recovery with $\ell_2/\ell_2$ guarantee, and show two results:
\begin{enumerate}
\item A scheme with $m=O(\frac{1}{\eps}k \log \log (n\eps/k))$
  measurements and an approximation factor $C=1+\eps$.  For low values
  of $k$ this provides an {\em exponential} improvement over the best
  possible non-adaptive bound. The scheme uses $O(\log^* k \cdot
  \log \log (n\eps/k))$ rounds.
\item A scheme with $m=O(\frac{1}{\eps}k \log (k/\eps) + k \log
  (n/k))$ and an approximation factor $C=1+\epsilon$.  For low values
  of $k$ and $\epsilon$ this offers a significant improvement over the
  best possible non-adaptive bound, since the dependence on $n$ and
  $\epsilon$ is ``split'' between two terms. The scheme uses only two
  rounds.
\end{enumerate}

\xxx{Composition?}

\paragraph{Implications}
Our new bounds lead to potentially significant improvements to
efficiency of sparse recovery schemes in a number of application
domains. Naturally, not {\em all} applications support adaptive
measurements. For example, network monitoring requires the
measurements to be performed simultaneously, since we cannot ask the
network to ``re-run'' the packets all over again. However, a
surprising number of applications are capable of supporting
adaptivity. For example:
\begin{itemize}
\item Streaming algorithms for data analysis: since each measurement
  round can be implemented by one pass over the data, adaptive schemes
  simply correspond to multiple-pass streaming algorithms (see ~\cite{McG} 
for some examples of such algorithms).
\item Compressed sensing of signals:  several architectures for
  compressive sensing, e.g., the single-pixel camera
  of~\cite{DDTLTKB}, already perform the measurements in a sequential
  manner. 
   In such cases the measurements can be made adaptive\footnote{We note that, in any realistic sensing system, minimizing the number of measurements is only one of several considerations. Other factors include: minimizing the computation time, minimizing the amount of communication needed to transfer the measurement matrices to the sensor, satisfying constraints on the measurement matrix imposed by the hardware etc. A detailed cost analysis covering all of these factors is architecture-specific, and beyond the scope of this paper. }.
   Other architectures supporting adaptivity are under development~\cite{KeCoM}.
  
\item Genetic data analysis and acqusition: as above.
\end{itemize}

Therefore, it seems likely that  the results in this paper will be applicable in a wide variety of scenarios.

\ifconf
\else
As an example application, we show how to solve the problem
of finding a duplicate in a data stream of $n$ arbitrarily chosen
items from the set $\{1, 2, \ldots, n-1\}$ presented in an arbitrary
order. Our algorithm uses $O(\log n)$
bits of space and $O(\log \log n)$ passes. It is known that
for a single pass, $\Theta(\log^2  n)$ bits of space is necessary
and sufficient \cite{jst11}, and so our algorithm improves upon
the best possible space complexity using a single pass. 
\fi

\paragraph{Techniques}
On a high-level, both of our schemes follow the same two-step process.
First, we reduce the problem of finding the best $k$-sparse
approximation to the problem of finding the best $1$-sparse
approximation (using relatively standard techniques). This is followed
by solving the latter (simpler) problem.

The first scheme starts by ``isolating'' most of of the large
coefficients by randomly sampling $\approx \eps/k$ fraction of the
coordinates; this mostly follows the approach of~\cite{GLPS}
(cf.~\cite{GGIKMS02:Fast-Small-Space}). The crux of the algorithm is
in the identification of the isolated coefficients. Note that in order
to accomplish this using $O(\log \log n)$ measurements (as opposed to
$O(\log n)$ achieved by the ``standard'' binary search algorithm) we
need to ``extract'' significantly more than one bit of information per
measurements.  To achieve this, we proceed as follows.  First, observe
that if the given vector (say, $z$) is {\em exactly} $1$-sparse, then
one can extract the position of the non-zero entry (say $z_j$) from
{\em two} measurements $a(z)=\sum_i z_i$, and $b(z)=\sum_i i z_i$,
since $b(z)/a(z)=j$. A similar algorithm works even if $z$ contains
some ``very small'' non-zero entries: we just round $b(z)/a(z)$ to the
nearest integer. This algorithm is a 
special case of a general algorithm that achieves $O(\log n / \log SNR)$
measurements to identify a single coordinate $x_j$ among $n$ coordinates, where
$SNR = x_j^2/\|x_{[n] \setminus j}\|^2$ (SNR stands for signal-to-noise ratio). 
This is optimal as a function of $n$ and the SNR \cite{DIPW}.

A natural approach would then be to partition $[n]$ into two sets $\{1, \ldots, n/2\}$ and 
$\{n/2+1, \ldots n\}$, find the heavier of the two sets, and recurse. This would 
take $O(\log n)$ rounds. The key observation is that not only do we 
recurse on a smaller-sized set of coordinates, but the SNR has also increased 
since $x_j^2$ has remained the same but the squared norm of the tail has dropped 
by a constant factor. Therefore in the next round we can afford to partition 
our set into more than two sets, since as long as we keep the ratio of 
$\log (\# \textrm{ of sets })$ and $\log SNR$ constant, we only need 
$O(1)$ measurements per round. This ultimately leads to a scheme that finishes 
after $O(\log \log n)$ rounds.


In the second scheme, we start by hashing the coordinates into a
universe of size polynomial in $k$ and $1/\epsilon$, in a way that
approximately preserves the top coefficients without introducing
spurious ones, and in such a way that the mass of the tail of the
vector does not increase significantly by hashing. This idea is
inspired by techniques in the data stream literature for estimating
moments \cite{knpw, TZ}
(cf. \cite{CCF,CM06:Combinatorial-Algorithms,GI}).  Here, though, we
need stronger error bounds. This enables us to identify the positions
of those coefficients (in the hashed space) using only
$O(\frac{1}{\eps}k \log (k/\eps))$ measurements. Once this is done,
for each large coefficient $i$ in the hash space, we identify the
actual large coefficient in the preimage of $i$. This can be achieved
using the number of measurements that does not depend on
$\epsilon$. \xxx{Perhaps even k (log n)/(log k) ?}

\section{Preliminaries}

We start from a few definitions. Let $x$ be an $n$-dimensional vector.

\begin{definition}
  Define
  \[
  H_k(x) = \argmax_{\substack{S \in [n]\\\abs{S} = k}} \norm{2}{x_S}
  \]
  to be the largest $k$ coefficients in $x$.
\end{definition}

\begin{definition}
  For any vector $x$, we define the ``heavy hitters'' to be those
  elements that are both (i) in the top $k$ and (ii) large relative to
  the mass outside the top $k$.  We define
  \[
  H_{k, \eps}(x) = \{j \in H_k(x) \mid x_j^2 \geq \eps \norm{2}{x_{\overline{H_k(x)}}}^2\}
  \]
\end{definition}
\begin{definition}
  Define the error
  \[
  \err(x, k) = \norm{2}{x_{\overline{H_k(x)}}}^2
  \]
\end{definition}

For the sake of clarity, the analysis of the algorithm in
section~\ref{s:tworound} assumes that the entries of $x$ are sorted by
the absolute value (i.e., we have $|x_1| \ge |x_2| \ge \ldots \ge
|x_n|$). In this case, the set $H_k(x)$ is equal to $[k]$; this allows
us to simplify the notation and avoid double subscripts. The
algorithms themselves are invariant under the permutation of the
coordinates of $x$.

\paragraph{Running times of the recovery algorithms} In the non-adaptive model, the running time of the recovery algorithm is well-defined: it is the number of operations performed by a procedure that takes $Ax$ as its input and produces an approximation $x^*$ to $x$. The time needed to generate the measurement vectors $A$, or to encode the vector $x$ using $A$, is not included. 
In the adaptive case, the distinction between the matrix generation, encoding and recovery procedures does not exist, since new measurements are generated based on the values of the prior ones.
Moreover, the running time of the measurement generation procedure
heavily depends on the representation of the matrix. If we suppose
that we may output the matrix in sparse form and receive encodings in
time bounded by the number of non-zero entries in the matrix, our
algorithms run in $n \log^{O(1)} n$ time.

\section{Full adaptivity}
\label{s:full}

This section shows how to perform $k$-sparse recovery with $O(k \log
\log (n/k))$ measurements.  The core of our algorithm is a method for
performing $1$-sparse recovery with $O(\log \log n)$ measurements.  We
then extend this to $k$-sparse recovery via repeated subsampling.

\subsection{$1$-sparse recovery}

This section discusses recovery of $1$-sparse vectors with $O(\log
\log n)$ adaptive measurements.  First, in
Lemma~\ref{lemma:nonadaptiveon} we show that if the heavy hitter $x_j$
is $\Omega(n)$ times larger than the $\ell_2$ error ($x_j$ is
``$\Omega(n)$-heavy''), we can find it with two non-adaptive
measurements.  This corresponds to non-adaptive $1$-sparse recovery with
approximation factor $C = \Theta(n)$; achieving this with $O(1)$
measurements is unsurprising, because the lower bound~\cite{DIPW} is
$\Omega(\log_{1 + C} n)$.

Lemma~\ref{lemma:nonadaptiveon} is not directly very useful, since
$x_j$ is unlikely to be that large.  However, if $x_j$ is $D$ times
larger than everything else, we can partition the coordinates of $x$
into $D$ random blocks of size $N/D$ and perform dimensionality
reduction on each block.  The result will in expectation be a vector
of size $D$ where the block containing $j$ is $D$ times larger than
anything else.  The first lemma applies, so we can recover the block
containing $j$, which has a $1/\sqrt{D}$ fraction of the $\ell_2$
noise.  Lemma~\ref{lemma:shrinknoise} gives this result.

We then have that with two non-adaptive measurements of a $D$-heavy
hitter we can restrict to a subset where it is an $\Omega(D^{3/2})$-heavy
hitter.  Iterating $\log \log n$ times gives the result, as shown in
Lemma~\ref{lemma:onesparse}.

\begin{lemma}\label{lemma:nonadaptiveon}
  Suppose there exists a $j$ with $\abs{x_j} \geq C
  \frac{n}{\sqrt{\delta}} \norm{2}{x_{[n]\setminus\{j\}}}$ for some
  constant $C$.  Then two non-adaptive measurements suffice to recover
  $j$ with probability $1-\delta$.
\end{lemma}
\begin{proof}
  Let $s \colon [n] \to \{\pm 1\}$ be chosen from a $2$-wise
  independent hash family.  Perform the measurements $a(x) = \sum s(i)x_i$
  and $b(x) = \sum (n + i) s(i)x_i$.  For recovery, output the closest
  integer to $b / a - n$.

  Let $z = x_{[n]\setminus \{j\}}$.  Then $\E[a(z)^2] = \norm{2}{z}^2$ and $\E[b(z)^2]
  \leq 4n^2\norm{2}{z}^2$.  Hence with probability at least $1 -
  2\delta$, we have both
  \begin{align*}
    \abs{a(z)} \leq \sqrt{1/\delta}\norm{2}{z}\\
    \abs{b(z)} \leq 2n\sqrt{1/\delta}\norm{2}{z}
  \end{align*}
  Thus
  \begin{align*}
    \frac{b(x)}{a(x)} =& \frac{s(j)(n+j)x_j + b(z)}{s(j)x_j + a(z)}\\
    \abs{\frac{b(x)}{a(x)} - (n + j)} =& \abs{\frac{b(z) - (n+j)a(z)}{s(j)x_j + a(z)}}\\
    \leq& \frac{\abs{b(z)} + (n+j)\abs{a(z)} }{\abs{\abs{x_j} - \abs{a(z)}}}\\
    \leq& \frac{4n\sqrt{1/\delta}\norm{2}{z}}{\abs{\abs{x_j} - \abs{a(z)}}}
  \end{align*}
  Suppose $\abs{x_j} > (8n+1)\sqrt{1/\delta}\norm{2}{z}$.  Then
  \begin{align*}
    \abs{\frac{b(x)}{a(x)} - (n + j)} <& \frac{4n\sqrt{1/\delta}\norm{2}{z}}{8n\sqrt{1/\delta}\norm{2}{z}}\\
    =& 1/2
  \end{align*}
  so $\hat{\imath} = j$.
\end{proof}

\begin{lemma}\label{lemma:shrinknoise}
  Suppose there exists a $j$ with $\abs{x_{j}} \geq
  C\frac{B^2}{\delta^2} \norm{2}{x_{[n]\setminus\{j\}}}$ for some
  constant $C$ and parameters $B$ and $\delta$.  Then with two
  non-adaptive measurements, with probability $1-\delta$ we can find a
  set $S \subset [n]$ such that $j \in S$ and $\norm{2}{x_{S \setminus
      \{j\}}} \leq \norm{2}{x_{[n]\setminus\{j\}}} / B$ and $\abs{S}
  \leq 1 + n/B^2$.
\end{lemma}

\begin{proof}
  Let $D = B^2/\delta$, and let $h \colon [n] \to [D]$ and $s \colon
  [n] \to \{\pm 1\}$ be chosen from pairwise independent hash
  families.  Then define $S_p = \{i \in [n] \mid h(i) = p\}$.  Define
  the matrix $A\in \R^{D\times n}$ by $A_{h(i),i} = s(i)$ and
  $A_{p,i}=0$ elsewhere.  Then
  \[
  (Az)_p = \sum_{i \in S_p} s(i)z_i.
  \]

  Let $p^* = h(j)$ and $y = x_{[n]\setminus\{j\}}$.  We have that
  \begin{align*}
    \E[\abs{S_{p^*}}] =& 1 + (n-1)/D\\
    \E[ (Ay)_{p^*}^2] = \E[\norm{2}{y_{S_{p^*}}}^2] =& \norm{2}{y}^2 / D\\
    \E[\norm{2}{Ay}^2] =& \norm{2}{y}^2
  \end{align*}
  Hence by Chebyshev's inequality, with probability at least $1 -
  4\delta$ all of the following hold:
  \begin{align}
    \abs{S_{p^*}} \leq& 1 + (n-1) / (D\delta) \leq 1 + n / B^2 \label{eq:Sp*}\\
    \norm{2}{y_{S_{p^*}}} \leq& \norm{2}{y} / \sqrt{D\delta} \label{eq:p*error}\\
    \abs{(Ay)_{p^*}} \leq& \norm{2}{y} / \sqrt{D\delta} \label{eq:Ap*error}\\
    \norm{2}{Ay} \leq& \norm{2}{y} /\sqrt{\delta}.\label{eq:Ayerror}
  \end{align}
  The combination of~\eqref{eq:Ap*error} and~\eqref{eq:Ayerror} imply
  \begin{align*}
    \abs{(Ax)_{p^*}} \geq& \abs{x_{j}} - \abs{(Ay)_{p^*}} \geq
    (CD/\delta - 1/\sqrt{D\delta}) \norm{2}{y}
    \geq
    (CD/\delta - 1/\sqrt{D\delta}) \sqrt{\delta}\norm{2}{Ay} \geq
    \frac{CD}{2\sqrt{\delta}} \norm{2}{Ay}
  \end{align*}
  and hence
  \[
  \abs{(Ax)_{p^*}}  \geq   \frac{CD}{2\sqrt{\delta}} \norm{2}{(Ax)_{[D] \setminus p^*}}.
  \]
  As long as $C/2$ is larger than the constant in
  Lemma~\ref{lemma:nonadaptiveon}, this means two non-adaptive
  measurements suffice to recover $p^*$ with probability $1 - \delta$.
  We then output the set $S_{p^*}$, which by~\eqref{eq:p*error} has
  \begin{align*}
    \norm{2}{x_{S_{p^*} \setminus \{j\}}} =& \norm{2}{y_{S_{p^*}}} \leq
    \norm{2}{y}/\sqrt{D\delta}
    =
    \norm{2}{x_{[n]\setminus\{j\}}} /
    \sqrt{D\delta} = \norm{2}{x_{[n]\setminus\{j\}}} / B
  \end{align*}
  as desired.  The overall failure probability is $1 - 5\delta$;
  rescaling $\delta$ and $C$ gives the result.
\end{proof}

\begin{algorithm}
  \caption{Adaptive $1$-sparse recovery}\label{alg:adaptive1}
  \begin{algorithmic}
    \Procedure{NonAdaptiveShrink}{$x$, $D$}\ \ \ \ \ \ \ \ \ \ \Comment{Find smaller set $S$ containing heavy coordinate $x_j$}
    \State For $i \in [n]$, $s_1(i) \gets \{\pm 1\}, h(i) \gets [D]$
    \State For $i \in [D]$, $s_2(i) \gets \{\pm 1\}$ 
    \State $a \gets \sum s_1(i) s_2(h(i))x_i$\Comment{Observation}
    \State $b \gets \sum s_1(i) s_2(h(i))x_i(D + h(i))$\Comment{Observation}
    \State $p^* \gets \textsc{Round}(b/a - D)$.
    \State \Return $\{j^* \mid h(j^*) = p^*\}$.
    \EndProcedure
  \end{algorithmic}
  \begin{algorithmic}
    \Procedure{AdaptiveOneSparseRec}{$x$}\Comment{Recover heavy coordinate $x_j$}
    \State $S \gets [n]$
    \State $B \gets 2$, $\delta \gets 1/4$
    \While{$\abs{S} > 1$}
    \State $S \gets \textsc{NonAdaptiveShrink}(x_S, 4B^2/\delta)$
    \State $B \gets B^{3/2}$, $\delta \gets \delta / 2$.
    \EndWhile
    \State \Return $S[0]$
    \EndProcedure
  \end{algorithmic}
\end{algorithm}

\begin{lemma}\label{lemma:onesparse}
  Suppose there exists a $j$ with $\abs{x_j} \geq C
  \norm{2}{x_{[n]\setminus\{j\}}}$ for some constant $C$.  Then
  $O(\log \log n)$ adaptive measurements suffice to recover $j$ with
  probability $1/2$.
\end{lemma}

\begin{proof}
  Let $C'$ be the constant from Lemma~\ref{lemma:shrinknoise}.  Define
  $B_0 = 2$ and $B_{i} = B_{i-1}^{3/2}$ for $i \geq 1$.  Define
  $\delta_i = 2^{-i}/4$ for $i \geq 0$.  Suppose $C \geq 16C'
  B_0^2/\delta_0^2$.

  Define $r = O(\log \log n)$ so $B_r \geq n$.  Starting with $S_0 =
  [n]$, our algorithm iteratively applies
  Lemma~\ref{lemma:shrinknoise} with parameters $B = 4B_i$ and $\delta
  = \delta_i$ to $x_{S_i}$ to identify a set $S_{i+1} \subset S_i$
  with $j \in S_{i+1}$, ending when $i = r$.

  We prove by induction that Lemma~\ref{lemma:shrinknoise} applies at
  the $i$th iteration.  We chose $C$ to match the base case.  For the
  inductive step, suppose $\norm{2}{x_{S_i \setminus \{j\}}}\leq
  \abs{x_j} / (C' 16\frac{B_i^2}{\delta_i^2})$.  Then by
  Lemma~\ref{lemma:shrinknoise},
  \[
  \norm{2}{x_{S_{i+1} \setminus \{j\}}} \leq \abs{x_j} / (C'
  64\frac{B_i^3}{\delta_i^2}) = \abs{x_j} / (C'
  16\frac{B_{i+1}^2}{\delta_{i+1}^2})
  \]
  so the lemma applies in the next iteration as well, as desired.

  After $r$ iterations, we have $S_r \leq 1 + n / B_r^2 < 2$, so we
  have uniquely identified $j \in S_r$.  The probability that any
  iteration fails is at most $\sum \delta_i < 2\delta_0 = 1/2$.
\end{proof}

\subsection{$k$-sparse recovery}

Given a $1$-sparse recovery algorithm using $m$ measurements, one can
use subsampling to build a $k$-sparse recovery algorithm using $O(km)$
measurements and achieving constant success probability.  Our method
for doing so is quite similar to one used in~\cite{GLPS}. The main
difference is that, in order to identify one large coefficient among a
subset of coordinates, we use the adaptive algorithm from the previous
section as opposed to error-correcting codes.

For intuition, straightforward subsampling at rate $1/k$ will, with
constant probability, recover (say) 90\% of the heavy hitters using
$O(km)$ measurements.  This reduces the problem to $k/10$-sparse
recovery: we can subsample at rate $10/k$ and recover 90\% of the
remainder with $O(km/10)$ measurements, and repeat $\log k$ times.
The number of measurements decreases geometrically, for $O(km)$ total
measurements.  Naively doing this would multiply the failure
probability and the approximation error by $\log k$; however, we can
make the number of measurements decay less quickly than the sparsity.
This allows the failure probability and approximation ratios to also
decay exponentially so their total remains constant.

To determine the number of rounds, note that the initial set of
$O(km)$ measurements can be done in parallel for each subsampling, so
only $O(m)$ rounds are necessary to get the first 90\% of heavy
hitters.  Repeating $\log k$ times would require $O(m \log k)$ rounds.
However, we can actually make the sparsity in subsequent iterations
decay super-exponentially, in fact as a power tower.  This give $O(m
\log^* k)$ rounds.

\begin{theorem}\label{thm:adaptiverecovery}
  There exists an adaptive $(1+\eps)$-approximate $k$-sparse recovery
  scheme with $O(\frac{1}{\eps}k\log \frac{1}{\delta}\log \log
  (n\eps/k))$ measurements and success probability $1-\delta$.  It
  uses $O(\log^* k \log \log (n \eps))$ rounds.
\end{theorem}

To prove this, we start from the following lemma:

\begin{lemma}\label{lemma:subsample}
  We can perform $O(\log \log (n/k))$ adaptive measurements and
  recover an $\hat{\imath}$ such that, for any $j \in H_{k, 1/k}(x)$
  we have $\Pr[\hat{\imath} = j] = \Omega(1/k)$.
\end{lemma}
\begin{proof}
  Let $S = H_k(x)$.  Let $T \subset [n]$ contain each element
  independently with probability $p = 1/(4C^2k)$, where $C$ is the
  constant in Lemma~\ref{lemma:onesparse}.  Let $j \in H_{k,1/k}(x)$.
  Then we have
  \[
  \E[\norm{2}{x_{T \setminus S}}^2] = p\norm{2}{x_{\overline{S}}}^2
  \]
  so
  \[
  \norm{2}{x_{T \setminus S}} \leq \sqrt{4p}\norm{2}{x_{\overline{S}}} =  \frac{1}{C\sqrt{k}}\norm{2}{x_{\overline{S}}} \leq \abs{x_j} / C
  \]
  with probability at least $3/4$. Furthermore we have $\E[\abs{T
      \setminus S}] < pn$ so $\abs{T \setminus S} < n/k$ with
  probability at least $1 - 1/(4C^2) > 3/4$.  By the union bound, both
  these events occur with probability at least $1/2$.

  Independently of this, we have
  \[
  \Pr[T \cap S = \{j\}] = p (1-p)^{k-1} > p/e
  \]
  so all these events hold with probability at least $p/(2e)$.  Assuming this,
  \[
  \norm{2}{x_{T \setminus \{j\}}} \leq \abs{x_j} / C
  \]
  and $\abs{T} \leq 1 + n/k$.  But then Lemma~\ref{lemma:onesparse}
  applies, and $O(\log \log \abs{T}) = O(\log \log (n/k))$
  measurements can recover $j$ from a sketch of $x_T$ with probability
  $1/2$.  This is independent of the previous probability, for a total
  success chance of $p / (4e) = \Omega(1/k)$.
\end{proof}

\begin{lemma}\label{lemma:recursion}
  With $O(\frac{1}{\eps}k \log \frac{1}{f\delta}\log \log (n\eps/k))$
  adaptive measurements, we can recover $T$ with $\abs{T} \leq k$ and
  \[
  \err(x_{\overline{T}}, f k) \leq (1 + \eps) \err(x, k)
  \]
  with probability at least $1-\delta$.  The number of rounds required
  is $O(\log \log (n \eps / k))$.
\end{lemma}

\begin{proof}
  Repeat Lemma~\ref{lemma:subsample} $m = O(\frac{1}{\eps}k \log
  \frac{1}{f \delta})$ times in parallel with parameters $n$ and
  $k/\eps$ to get coordinates $T' = \{t_1, t_2, \dotsc, t_m\}$.  For
  each $j \in H_{k,\eps/k}(x) \subseteq H_{k/\eps,\eps/k}(x)$ and $i
  \in [m]$, the lemma implies $\Pr[j = t_i] \geq \eps/(Ck)$ for some
  constant $C$.  Then $\Pr[j \notin T'] \leq (1 - \eps/(Ck))^m \leq
  e^{-\eps m/(Ck)} \leq f\delta$ for appropriate $m$.  Thus
  \begin{align*}
    \E[\abs{H_{k,\eps/k}(x) \setminus T'}] \leq f\delta\abs{H_{k,\eps/k}(x)} \leq& f\delta k\\
    \Pr\left[\abs{H_{k,\eps/k}(x) \setminus T'} \geq f k\right] \leq& \delta.
  \end{align*}

  Now, observe $x_{T'}$ directly and set $T \subseteq T'$ to be the
  locations of the largest $k$ values.  Then, since $H_{k,\eps/k}(x)
  \subseteq H_k(x)$, $\abs{H_{k,\eps/k}(x) \setminus T} =
  \abs{H_{k,\eps/k}(x) \setminus T'} \leq f k$ with probability at
  least $1-\delta$.

  Suppose this occurs, and let $y = x_{\overline{T}}$.  Then
  \begin{align*}
    \err(y, f k) =& \min_{\abs{S} \leq f k} \norm{2}{y_{\overline{S}}}^2\\
    \leq& \norm{2}{y_{\overline{H_{k,\eps/k}(x) \setminus T}}}^2\\
    =& \norm{2}{x_{\overline{H_{k,\eps/k}(x)}}}^2\\
    =& \norm{2}{x_{\overline{H_k(x)}}}^2 + \norm{2}{x_{H_k(x) \setminus H_{k,\eps/k}(x)}}^2\\
    \leq& \norm{2}{x_{\overline{H_k(x)}}}^2 + k\norm{\infty}{x_{H_k(x) \setminus H_{k,\eps/k}(x)}}^2\\
    \leq& (1 + \eps)\norm{2}{x_{\overline{H_k(x)}}}^2\\
    =& (1+\eps)\err(x, k)
  \end{align*}
  as desired.
\end{proof}

\begin{algorithm}
  \caption{Adaptive $k$-sparse recovery}\label{alg:adaptivek}
  \begin{algorithmic}
    \Procedure{AdaptiveKSparseRec}{$x$, $k$, $\eps$, $\delta$}\Comment{Recover approximation $\hat{x}$ of $x$}
    \State $R_0 \gets [n]$
    \State $\delta_0 \gets \delta / 2$, $\eps_0 \gets \eps / e$, $f_0 \gets 1/32$, $k_0 \gets k$.
    \State $J \gets \{\}$
    \For{$i \gets 0, \dotsc, O(\log^* k)$}\Comment{While $k_i \geq 1$}
    \For{$t \gets 0, \dotsc, \Theta(\frac{1}{\eps_i}k_i\log \frac{1}{\delta_i})$}
    \State $S_t \gets \textsc{Subsample}(R_i, \Theta(\eps_i/k_i))$
    \State $J.\text{add}(\textsc{AdaptiveOneSparseRec}(x_{S_t}))$
    \EndFor
    \State $R_{i+1} \gets [n] \setminus J$
    \State $\delta_{i+1} \gets \delta_i / 8$
    \State $\eps_{i+1} \gets \eps_i / 2$
    \State $f_{i+1} \gets 1/2^{1/(4^{i+1}f_i)}$
    \State $k_{i+1} \gets k_i f_i$
    \EndFor
    \State $\hat{x} \gets x_J$ \Comment{Direct observation}
    \State \Return $\hat{x}$
    \EndProcedure
  \end{algorithmic}
\end{algorithm}

\begin{theorem}
  We can perform $O(\frac{1}{\eps}k \log \frac{1}{\delta} \log \log
  (n\eps/k))$ adaptive measurements and recover a set $T$ of size at most $2k$
  with
  \[
  \norm{2}{x_{\overline{T}}} \leq (1+\eps) \norm{2}{x_{\overline{H_k(x)}}}.
  \]
  with probability $1-\delta$.  The number of rounds required is
  $O(\log^* k \log \log (n\eps))$.
\end{theorem}
\begin{proof}
  Define $\delta_i = \frac{\delta}{2 \cdot 2^i}$ and $\eps_i =
  \frac{\eps}{e\cdot 2^i}$.  Let $f_0 = 1/32$ and $f_i = 2^{-1 / (4^i f_{i-1})}$
  for $i > 0$, and define $k_i = k\prod_{j < i} f_j$.  Let $R_0 = [n]$.

  Let $r = O(\log^* k)$ such that $f_{r-1} < 1/k$.  This is possible
  since $\alpha_i = 1/(4^{i+1} f_i)$ satisfies the recurrence
  $\alpha_0 = 8$ and $\alpha_i = 2^{\alpha_{i-1} - 2i - 2} >
  2^{\alpha_{i-1} / 2}$.  Thus $\alpha_{r-1} > k$ for $r = O(\log^*
  k)$ and then $f_{r-1} < 1/\alpha_{r-1} < 1/k$.

  For each round $i = 0, \dotsc, r-1$, the algorithm runs
  Lemma~\ref{lemma:recursion} on $x_{R_i}$ with parameters $\eps_i$,
  $k_i$, $f_i$, and $\delta_i$ to get $T_i$.  It sets $R_{i+1} = R_i
  \setminus T_i$ and repeats.  At the end, it outputs $T = \cup T_i$.

  The total number of measurements is
  \begin{align*}
    O(\sum  \frac{1}{\eps_i} k_i \log \frac{1}{f_i\delta_i} \log \log (n \eps_i/ k_i))
    \leq& O(\sum \frac{2^i (k_i/k)\log (1/f_i)}{\eps} k (i + \log \frac{1}{\delta})
 \log (\log (k/k_i) + \log (n\eps/k)))\\
    \leq& O(\frac{1}{\eps} k \log \frac{1}{\delta} \log \log (n\eps/k)
\sum 2^i (k_i/k) \log (1/f_i) (i+1) \log \log (k/k_i))
   \end{align*}
   using the very crude bounds $i + \log (1/\delta) \leq
   (i+1)\log(1/\delta)$ and $\log (a + b) \leq 2\log a \log b$ for $a,
   b \geq e$.  But then
   \begin{align*}
     \sum 2^i (k_i/k) \log (1/f_i) (i+1) \log \log (k/k_i) 
\leq& \sum 2^i (i+1)  f_i \log (1/f_i) \log \log (1/f_i)\\
     \leq& \sum 2^i (i+1) O(\sqrt{f_i})\\
     =& O(1)
  \end{align*}
  since $f_i < O(1/16^i)$, giving $O(\frac{1}{\eps} k \log
  \frac{1}{\delta} \log \log (n\eps/k)$ total measurements.  The
  probability that any of the iterations fail is at most $\sum
  \delta_i < \delta$.  The result has size $\abs{T} \leq \sum k_i \leq
  2k$.  All that remains is the approximation ratio
  $\norm{2}{x_{\overline{T}}} = \norm{2}{x_{R_r}}$.

  For each $i$, we have
  \begin{align*}
    \err(x_{R_{i+1}}, k_{i+1}) =& \err(x_{R_i \setminus T_i}, f_i k_i)
    \leq
    (1 + \eps_i) \err(x_{R_i}, k_i).
  \end{align*}
  Furthermore, $k_r < k f_{r-1} < 1$.  Hence
  \begin{align*}
    \norm{2}{x_{R_r}}^2 = \err(x_{R_r}, k_r) \leq&
    \left(\prod_{i=0}^{r-1} (1 + \eps_i)\right)\err(x_{R_0}, k_0)
    =
    \left(\prod_{i=0}^{r-1} (1 + \eps_i)\right)\err(x, k)
  \end{align*}
  But $\prod_{i=0}^{r-1} (1 + \eps_i) < e^{\sum \eps_i} < e$, so
  \[
  \prod_{i=0}^{r-1} (1 + \eps_i) < 1 + \sum e\eps_i \leq 1 + 2\eps
  \]
  and hence
  \[
  \norm{2}{x_{\overline{T}}}  = \norm{2}{x_{R_r}} \leq (1+\eps) \norm{2}{x_{\overline{H_k(x)}}}
  \]
  as desired.
\end{proof}

Once we find the support $T$, we can observe $x_T$ directly with
$O(k)$ measurements to get a $(1+\eps)$-approximate $k$-sparse recovery
scheme, proving Theorem~\ref{thm:adaptiverecovery}

\section{Two-round adaptivity}
\label{s:tworound}

The algorithms in this section are invariant under permutation.
Therefore, for simplicity of notation, the analysis assumes our
vectors $x$ is sorted: $\abs{x_1} \geq \dotsc \geq \abs{x_n} = 0$.

We are given a $1$-round $k$-sparse recovery algorithm for
$n$-dimensional vectors $x$ using $m(k,\eps, n, \delta)$ measurements
with the guarantee that its output $\hat{x}$ satisfies
$\|\hat{x}-x\|_p \leq (1+\eps) \cdot \|x_{\overline{[k]}}\|_p$ for a
$p \in \{1,2\}$ with probability at least $1-\delta$.  Moreover,
suppose its output $\hat{x}$ has support on a set of size $s(k,\eps,
n, \delta)$.  We show the following black box two-round
transformation.
\begin{theorem}\label{thm:tworound}
Assume $s(k, \eps, n, \delta) = O(k)$.  Then there is a $2$-round
sparse recovery algorithm for $n$-dimensional vectors $x$, which, in
the first round uses $m(k, \eps/5, \poly(k/\eps), 1/100)$ measurements
and in the second uses $ O(k \cdot m(1, 1, n, \Theta(1/k)))$ measurements.
It succeeds with constant probability.
\end{theorem}

\begin{corollary}\label{cor:tworound}
  For $p = 2$, there is a $2$-round sparse recovery algorithm for 
$n$-dimensional
  vectors $x$ such that the total number of measurements is
  $O(\frac{1}{\eps}k \log(k/\eps) + k \log (n/k))$.
\end{corollary}
\begin{proof}[Proof of Corollary~\ref{cor:tworound}.]
In the first round it suffices to use CountSketch with 
$s(k, \eps, n, 1/100) = 2k$, which 
holds for
any $\eps > 0$ \cite{PW11}. We also have that
$m(k, \eps/5, \poly(k/\eps), 1/100) = O(\frac{1}{\eps} k \log(k/\eps))$. 
%
Using~\cite{CCF,CM06:Combinatorial-Algorithms,GI}, in the second round
we can set $m(1, 1, n, \Theta(1/k)) = O( \log n)$.  The bound follows
by observing that $\frac{1}{\eps}k \log(k/\eps) + k \log (n) =
O(\frac{1}{\eps}k \log(k/\eps) + k \log (n/k))$.
\end{proof}

\begin{proof}[Proof of Theorem~\ref{thm:tworound}.]
In the first round we perform a dimensionality reduction of the
$n$-dimensional input vector $x$ to a $\poly(k/\eps)$-dimensional
input vector $y$. We then apply the black box sparse recovery
algorithm on the reduced vector $y$, obtaining a list of $s(k, \eps/5,
\poly(k/\eps), 1/100)$ coordinates, and show for each coordinate in
the list, if we choose the largest preimage for it in $x$, then this
list of coordinates can be used to provide a $1+\eps$ approximation
for $x$.  In the second round we then identify which heavy coordinates
in $x$ map to those found in the first round, for which it suffices to
invoke the black box algorithm with only a constant approximation. We
place the estimated values of the heavy coordinates obtained in the
first pass in the locations of the heavy coordinates obtained in the
second pass.

Let $N = \poly(k/\eps)$ be determined below. Let $h:[n] \rightarrow
[N]$ and $\sigma: [n] \rightarrow \{-1,1\}$ be $\Theta(\log N)$-wise
independent random functions.  Define the vector $y$ by $y_i = \sum_{j
  \ \mid \ h(j) = i} \sigma(j) x_j$.  Let $Y(i)$ be the vector $x$
restricted to coordinates $j \in [n]$ for which $h(j) = i$.  Because
the algorithm is invariant under permutation of coordinates of $y$, we
may assume for simplicity of notation that $y$ is sorted: $\abs{y_1}
\geq \dotsc \geq \abs{y_N} = 0$.

We note that such a dimensionality reduction is often used in the
streaming literature. For example, the sketch of~\cite{TZ} for
$\ell_2$-norm estimation utilizes such a mapping.  A ``multishot''
version (that uses several functions $h$) has been used before in the
context of sparse recovery~\cite{CCF,CM06:Combinatorial-Algorithms}
(see~\cite{GI} for an overview).  Here, however, we need to analyze a
``single-shot'' version.

Let $p \in \{1,2\}$, and consider sparse recovery with the
$\ell_p/\ell_p$ guarantee.  We can assume that $\|x\|_p = 1$. We need
two facts concerning concentration of measure.

\begin{fact}\label{fact:chernoff}(see, e.g., Lemma 2 of \cite{knpw})
Let $X_1,\ldots,X_n$ be such that $X_i$ has expectation $\mu_i$ and
variance $v_i^2$, and $X_i \le K$ almost surely.  Then if the
$X_i$ are $\ell$-wise independent for an even integer $\ell\ge 2$,
\[
\Pr\left[\left|\sum_{i=1}^n X_i - \mu\right| \ge \lambda\right] \le
2^{O(\ell)} \left(\left(v \sqrt{\ell}/\lambda\right)^\ell +
  \left(K\ell/\lambda\right)^\ell\right)
\]
where $\mu = \sum_i \mu_i$ and $v^2 = \sum_i v_i^2$.
\end{fact}

\begin{fact}\label{fact:khintchine} (Khintchine inequality) (\cite{h82})
For $t \geq 2$, a vector $z$ and a $t$-wise independent random sign
vector $\sigma$ of the same number of dimensions, ${\bf E}[|\langle z,
  \sigma \rangle|^t] \leq \|z\|_2^t (\sqrt{t})^t.$
\end{fact}

We start with a probabilistic lemma. Let $Z(j)$ denote the vector
$Y(j)$ with the coordinate $m(j)$ of largest magnitude removed.
\begin{lemma}\label{lem:3cond}
Let $r = O \left (\|x_{\overline{[k]}}\|_p \cdot \frac{\log N}{N^{1/6}} \right)$ and $N$ be
sufficiently large. Then with probability $\geq 99/100$,
\begin{enumerate}
\item $\forall j \in [N]$, $\|Z(j)\|_p \leq r$.
\item $\forall i \in [N^{1/3}]$,
$|\sigma(i) \cdot y_{h(i)}- x_i| \leq r$,
\item $\|y_{\overline{[k]}}\|_p \leq (1 + O(1/\sqrt{N})) \cdot \|x_{\overline{[k]}}\|_p + O(kr)$,
\item $\forall j \in [N]$, if $h^{-1}(j) \cap [N^{1/3}] = \emptyset$,
then $|y_j| \leq r$,
\item $\forall j \in [N]$, $\|Y(j)\|_0 = O(n/N + \log N)$.
\end{enumerate}
\end{lemma}
\begin{proof}
We start by defining events $\mathcal{E}$, $\mathcal{F}$ and
$\mathcal{G}$ that will be helpful in the analysis, and showing that
all of them are satisfied simultaneously with constant probability.
\\\\
{\bf Event $\mathcal{E}$:}
Let $\mathcal{E}$ be the event that $h(1), h(2), \ldots, h(N^{1/3})$ are distinct.
Then $\Pr_{h}[\mathcal{E}] \geq 1-1/N^{1/3}$.
\\\\
{\bf Event $\mathcal{F}$:}
Fix $i \in [N]$. Let $Z'$ denote the vector $Y(h(i))$ with the coordinate $i$ removed.
Applying Fact \ref{fact:khintchine} with $t = \Theta(\log N)$,
\begin{eqnarray*}
\Pr_{\sigma}[|\sigma(i) y_{h(i)} - x_i| \geq 2 \sqrt{t} \cdot \|Z(h(i))\|_2]
& \leq & \Pr_{\sigma}[|\sigma(i) y_{h(i)} - x_i| \geq 2 \sqrt{t} \cdot \norm{2}{Z'}] \\
& \leq & \Pr_{\sigma}[|\sigma(i)y_{h(i)} - x_i|^t \geq 2^t (\sqrt{t})^t \cdot \|Z'\|_2^t]\\
& \leq & \Pr_{\sigma}[|\sigma(i)y_{h(i)} - x_i|^t \geq 2^t \E[\abs{\langle \sigma, Z'\rangle}^t]]\\
& = & \Pr_{\sigma}[|\sigma(i)y_{h(i)}-x_i|^t \geq 2^t \E[|\sigma(i)y_{h(i)}-x_i|^t] \leq 1/N^{4/3}.
\end{eqnarray*}
Let $\mathcal{F}$ be the event that for all $i \in [N]$,
$|\sigma(i)y_{h(i)}-x_i| \leq 2 \sqrt{t} \cdot \|Z(h(i))\|_2,$
so $\Pr_{\sigma}[\mathcal{F}] \geq 1-1/N$.
\\\\
{\bf Event $\mathcal{G}$:} 
Fix $j \in [N]$ and for each $i \in \{N^{1/3}+1, \ldots, n\}$, let
$X_i = |x_i|^p {\bf 1}_{h(i) = j}$ (i.e., $X_i=|x_i|^p$ if $h(i)=j$).
We apply Lemma \ref{fact:chernoff} to the $X_i$.  In the notation of
that lemma, $\mu_i = |x_i|^p/N$ and $v_i^2 \leq |x_i|^{2p}/N$, and so
$\mu = \|x_{\overline{[N^{1/3}]}}\|_p^p/N$ and $v^2 \leq
\|x_{\overline{[N^{1/3}]}}\|_{2p}^{2p}/N$. Also, $K =
|x_{N^{1/3}+1}|^p$.  Function $h$ is $\Theta(\log N)$-wise
independent, so by Fact \ref{fact:chernoff},
\begin{align*}
  \Pr \left [\left| \sum_i X_i -
      \frac{\|x_{\overline{[N^{1/3}]}}\|_p^p}{N} \right| \geq \lambda
  \right ]
  \leq& 2^{O(\ell)} \left ( \left
      (\|x_{\overline{[N^{1/3}]}}\|_{2p}^p  \sqrt{\ell}/(\lambda
       \sqrt{N}) \right )^{\ell} + \left (|x_{N^{1/3}+1}|^p 
      \ell / \lambda \right )^{\ell} \right )
\end{align*}
for any $\lambda > 0$ and an $\ell = \Theta(\log N)$. For $\ell$ large enough, there is a 
\[
\lambda = \Theta(\|x_{\overline{[N^{1/3}]}}\|_{2p}^p \sqrt{(\log
  N)/N} + |x_{N^{1/3}+1}|^p \cdot \log N)
\]
for which this probability is $\leq N^{-2}$. Let $\mathcal{G}$ be the
event that for all $j \in [N]$, $\|Z(j)\|_p^p \le C \left
(\frac{\|x_{\overline{[N^{1/3}]}}\|_p^p}{N} + \lambda \right )$ for
some universal constant $C>0$.  Then $\Pr[\mathcal{G} \mid
  \mathcal{E}] \geq 1-1/N$.
\\\\
By a union bound, $\Pr[\mathcal{E} \wedge \mathcal{F} \wedge
  \mathcal{G}] \geq 999/1000$ for $N$ sufficiently large.

We know proceed to proving the five conditions in the lemma statement.
In the analysis we assume that the event $\mathcal{E} \wedge
\mathcal{F} \wedge \mathcal{G}$ holds (i.e., we condition on that
event).
\\\\
{\bf First Condition:} This condition follows from the occurrence of
$\mathcal{G}$, and using that $\|x_{\overline{[N^{1/3}]}}\|_{2p} \leq
\|x_{\overline{[N^{1/3}]}}\|_p$, and $\|x_{\overline{[N^{1/3}]}}\|_p
\leq \|x_{\overline{[k]}}\|_p$, as well as
$(N^{1/3}-k+1)|x_{N^{1/3}+1}|^p \leq \|x_{\overline{[k]}}\|_p^p$. One
just needs to make these substitutions into the variable $\lambda$
defining $\mathcal{G}$ and show the value $r$ serves as an upper bound
(in fact, there is a lot of room to spare, e.g., $r/\log N$ is also an
upper bound).
\\\\
{\bf Second Condition:} This condition follows from the joint
occurrence of $\mathcal{E}$, $\mathcal{F}$, and $\mathcal{G}$.
\\\\
{\bf Third Condition: } For the third condition, let $y'$ denote the
restriction of $y$ to coordinates in the set $[N] \setminus \{h(1),
h(2), ..., h(k)\}$.  For $p = 1$ and for any choice of $h$ and
$\sigma$, $\|y'\|_1 \leq \|x_{\overline{[k]}}\|_1$.  For $p = 2$, the
vector $y$ is the sketch of~\cite{TZ} for $\ell_2$-estimation. By
their analysis, with probability $\geq 999/1000$, $\|y'\|_2^2 \leq
(1+O(1/\sqrt{N}))\|x'\|_2^2$, where $x'$ is the vector whose support
is $[n] \setminus \cup_{i=1}^k h^{-1}(i) \subseteq [n] \setminus [k]$.
We assume this occurs and add $1/1000$ to our error probability.
Hence, $\|y'\|_2^2 \leq (1+O(1/\sqrt{N}))\|x_{\overline{[k]}}\|_2^2$.

We relate $\|y'\|_p^p$ to $\|y_{\overline{[k]}}\|_p^p$.  Consider any
$j = h(i)$ for an $i \in [k]$ for which $j$ is not among the top $k$
coordinates of $y$. Call such a $j$ {\it lost}.  By the first
condition of the lemma, $|\sigma(i) y_j- x_i| \leq r$.  Since $j$ is
not among the top $k$ coordinates of $y$, there is a coordinate $j'$
among the top $k$ coordinates of $y$ for which $j' \notin h([k])$ and
$|y_{j'}| \geq |y_j| \geq |x_i| - r.$ We call such a $j'$ a {\it
  substitute}.  We can bijectively map substitutes to lost
coordinates. It follows that
$\|y_{\overline{[k]}}\|_p^p \leq \|y'\|_p^p
+ O(kr) \leq (1 + O(1/\sqrt{N})) \|x_{\overline{[k]}}\|_p^p + O(kr).$
\\\\
{\bf Fourth Condition:} This follows from the joint occurrence of
$\mathcal{E}, \mathcal{F}$, and $\mathcal{G}$, and using that
$|x_{m(j)}|^p \leq \|x_{\overline{[k]}}\|_p^p/(N^{1/3}-k+1)$ since
$m(j) \notin [N^{1/3}]$.
\\\\
{\bf Fifth Condition:} For the fifth condition, fix $j \in [N]$.  We
apply Fact \ref{fact:chernoff} where the $X_i$ are indicator variables
for the event $h(i) = j$. Then ${\bf E}[X_i] = 1/N$ and ${\bf
  Var}[X_i] < 1/N$.  In the notation of Fact \ref{fact:chernoff}, $\mu
= n/N$, $v^2 < n/N$, and $K = 1$. Setting $\ell = \Theta(\log N)$ and
$\lambda = \Theta(\log N + \sqrt{(n \log N)/N})$, we have by a union
bound that for all $j \in [N]$, $\|Y(j)\|_0 \leq \frac{n}{N} +
\Theta(\log N + \sqrt{(n \log N)/N}) = O(n/N + \log N)$, with
probability at least $1-1/N$.
\\\\
By a union bound, all events jointly occur with probability at least
$99/100$, which completes the proof.
\end{proof}

\noindent {\bf Event $\mathcal{H}$:} Let $\mathcal{H}$ be the event
that the algorithm returns a vector $\hat{y}$ with $\|\hat{y}-y\|_p
\leq (1+\eps/5)\|y_{\overline{[k]}}\|_p.$ Then $\Pr[\mathcal{H}] \geq
99/100$.  Let $S$ be the support of $\hat{y}$, so $|S|= s(k, \eps/5,
N, 1/100)$. We condition on $\mathcal{H}$.
\\\\
In the second round we run the algorithm on $Y(j)$ for each $j \in S$,
each using $m(1, 1, \|Y(j)\|_0, \Theta(1/k)))$- measurements. Using
the fifth condition of Lemma \ref{lem:3cond}, we have that $\|Y(j)\|_0
= O(\eps n/k + \log(k)/\eps)$ for $N = \poly(k/\eps)$ sufficiently
large.

For each invocation on a vector $Y(j)$ corresponding to a $j \in S$,
the algorithm takes the largest (in magnitude) coordinate HH$(j)$ in
the output vector, breaking ties arbitrarily. We output the vector
$\hat{x}$ with support equal to $T = \{\textrm{HH}(j) \mid j \in S
\}$. We assign the value $\sigma(x_j) \hat{y}_j$ to HH$(j)$. We have
\begin{align}\label{eqn:main} 
\|x-\hat{x}\|_p^p =& \|(x-\hat{x})_T\|_p^p + \|(x-\hat{x})_{[n] \setminus T}\|_p^p
=
\|(x-\hat{x})_T\|_p^p + \|x_{[n] \setminus T}\|_p^p.
\end{align}
The rest of the analysis is devoted to bounding the RHS of equation
\ref{eqn:main}.

\begin{lemma}
For $N = \poly(k/\eps)$ sufficiently large, 
conditioned on the events of Lemma \ref{lem:3cond} and $\mathcal{H}$,
$$\|x_{[n] \setminus T}\|_p^p \leq (1+
\eps/3)\|x_{\overline{[k]}}\|_p^p.$$
\end{lemma}
\begin{proof}
If $[k] \setminus T = \emptyset$, the lemma follows by
definition. Otherwise, if $i \in ([k] \setminus T)$, then $i \in [k]$,
and so by the second condition of Lemma \ref{lem:3cond}, $|x_i | \leq
|y_{h(i)}| + r.$ We also use the third condition of Lemma
\ref{lem:3cond} to obtain $\|y_{\overline{[k]}}\|_p \leq (1+
O(1/\sqrt{N})) \cdot \|x_{\overline{[k]}}\|_p + O(kr).$ By the
triangle inequality,
\begin{align*}
\left (\sum_{i \in [k] \setminus T} |x_i|^p \right )^{1/p} & \leq
k^{1/p}r + \left (\sum_{i \ \in \ [k] \setminus T} |y_{h(i)}|^p \right )^{1/p}
\leq
k^{1/p}r+ \left (\sum_{i \ \in \ [N] \setminus S} |y_{i}|^p \right )^{1/p}\\
\leq& k^{1/p}r + (1+\eps/5) \cdot \|y_{\overline{[k]}}\|_p.
\end{align*}
The lemma follows using that $r = O(\|x_{\overline{[k]}}\|_2 \cdot
(\log N)/N^{1/6})$ and $N = \poly(k/\eps)$ is sufficiently large.
\end{proof}
\noindent We bound $\|(x-\hat{x})_T\|_p^p$ using Lemma \ref{lem:3cond}, $|S| \leq \poly(k/\eps)$, and
that $N = \poly(k/\eps)$ is sufficiently large. 
\begin{align*}
\|(x-\hat{x})_T\|_p 
\leq& \left (\sum_{j \in S} |x_{HH(j)} - \sigma(HH(j)) \cdot \hat{y}_j|^p \right )^{1/p}
\leq
\left (\sum_{j \in S} (|y_j - \hat{y}_j| + |\sigma(HH(j)) \cdot x_{HH(j)} - y_j| )^p \right )^{1/p}\\
\leq& \left (\sum_{j \in S} |y_j - \hat{y}_j|^p \right )^{1/p}
 + \left (\sum_{j \in S}|\sigma(HH(j)) \cdot x_{HH(j)} - y_j|^p \right )^{1/p}\\
\leq& (1+\eps/5) \|y_{\overline{[k]}}\|_p 
+ \left (\sum_{j \in S} |\sigma(HH(j)) \cdot x_{HH(j)} - y_j|^p \right )^{1/p} \\
\leq& (1+\eps/5) (1+O(1/\sqrt{N}))\|x_{\overline{[k]}}\|_p + O(kr) 
+ \left (\sum_{j \in S} |\sigma(HH((j)) \cdot x_{HH(j)} - y_j|^p \right )^{1/p}  \\
\leq& (1+\eps/4) \|x_{\overline{[k]}}\|_p 
+ \left (\sum_{j \in S} |\sigma(HH(j)) \cdot x_{HH(j)} - y_j|^p \right )^{1/p}
\end{align*}
\noindent {\bf Event $\mathcal{I}$:} We condition on the event
$\mathcal{I}$ that all second round invocations succeed. Note that
$\Pr[\mathcal{I}] \geq 99/100$.
\\\\
We need the following lemma concerning $1$-sparse recovery algorithms. 
\begin{lemma}\label{lem:bitTest}
Let $w$ be a vector of real numbers. Suppose $|w_1|^p > \frac{9}{10} \cdot \|w\|^p_p$. 
Then for any vector $\hat{w}$ for which $\|w-\hat{w}\|^p_p \leq 2 
\cdot \|w_{\overline{[1]}}\|^p_p$, 
we have $|\hat{w}_{1}|^p > \frac{3}{5} \cdot \|w\|_p^p$.
Moreover, for all $j > 1$, $|\hat{w}_j|^p < \frac{3}{5} \cdot \|w\|_p^p$. 
\end{lemma}
\begin{proof}
$\|w-\hat{w}\|^p_p \geq |w_1-\hat{w}_1|^p$, so if 
$|\hat{w}_1|^p < \frac{3}{5} \cdot \|w\|^p_p$, then 
$\|w-\hat{w}\|^p_p > \left (\frac{9}{10} - \frac{3}{5} \right )\|w\|^p_p
= \frac{3}{10} \cdot \|w\|^p_p$.  
On the other hand,
$\|w_{\overline{[1]}}\|_p^p < \frac{1}{10} \cdot \|w\|_p^p$. This contradicts
that $\|w-\hat{w}\|^p_p \leq 2 \cdot \|w_{\overline{[1]}}\|_p^p$.
For the second part, for $j > 1$ we have
$|w_j|^p < \frac{1}{10} \cdot \|w\|^p_p$. Now,
$\|w-\hat{w}\|^p_p \geq |w_j - \hat{w}_j|^p$, so if
$|\hat{w}_j|^p \geq \frac{3}{5} \cdot \|w\|_p^p$, then
$\|w-\hat{w}\|^p_p > \left (\frac{3}{5} - \frac{1}{10} \right ) \|w\|_p^p
= \frac{1}{2} \cdot \|w\|_p^p$. But since
$\|w_{\overline{[1]}}\|_p^p < \frac{1}{10} \cdot \|w\|_p^p$, this contradicts that
$\|w-\hat{w}\|^p_p \leq 2\cdot \|w_{\overline{[1]}}\|_p^p$. 
\end{proof}
It remains to bound $\sum_{j \in S} |\sigma(HH(j)) \cdot x_{HH(j)}-y_j|^p$. 
We show for every $j \in S$, $|\sigma(HH(j)) \cdot x_{HH(j)}-y_j|^p$ is small.

Recall that $m(j)$ is the coordinate of $Y(j)$ with the largest magnitude.
There are two cases.
\\\\
{\bf Case 1: $m(j) \notin [N^{1/3}]$.} In this case observe that $HH(j) \notin [N^{1/3}]$
either, and $h^{-1}(j) \cap [N^{1/3}] = \emptyset$. 
It follows by the fourth condition of Lemma \ref{lem:3cond} that $|y_j| \leq r$. Notice
that $|x_{HH(j)}|^p \leq |x_{m(j)}|^p \leq \frac{\|x_{\overline{[k]}}\|_p^p}{N^{1/3}-k}.$
Bounding $|\sigma(HH(j)) \cdot x_{HH(j)}-y_j|$ by $|x_{HH(j)}| + |y_j|$, it follows for $N = \poly(k/\eps)$ 
large enough that $|\sigma(HH(j)) \cdot x_{HH(j)}-y_j|^p \le \eps/4 \cdot \|x_{\overline{[k]}}\|_p/|S|)$. 
\\\\
{\bf Case 2: $m(j) \in [N^{1/3}]$.}  If  $HH(j) = m(j)$, then
$|\sigma(HH(j)) \cdot x_{HH(j)}-y_j| \leq r$ by the second condition of Lemma \ref{lem:3cond}, and therefore 

$$|\sigma(HH(j)) \cdot x_{HH(j)}-y_j|^p \leq r^p \le \eps/4 \cdot \|x_{\overline{[k]}}\|_p/|S|$$ for $N = \poly(k/\eps)$
large enough. 

Otherwise, $HH(j) \neq m(j)$. From condition 2 of Lemma
\ref{lem:3cond} and $m(j) \in [N^{1/3}]$, it follows that
\begin{align*}
  |\sigma(HH(j))) x_{HH(j)}-y_j|
  \leq& |\sigma(HH(j)) x_{HH(j)} - \sigma(m(j)) x_{m(j)}| +
  |\sigma(m(j)) x_{m(j)} - y_j|
\leq
  |x_{HH(j)}| + |x_{m(j)}| + r
\end{align*}
Notice that $|x_{HH(j)}|+|x_{m(j)}| \leq 2|x_{m(j)}|$ since $m(j)$ is
the coordinate of largest magnitude. Now, conditioned on
$\mathcal{I}$, Lemma \ref{lem:bitTest} implies that $|x_{m(j)}|^p \leq
\frac{9}{10} \cdot \|Y(j)\|_p^p$, or equivalently, $|x_{m(j)}| \leq
10^{1/p} \cdot \|Z(j)\|_p.$ Finally, by the first condition of Lemma
\ref{lem:3cond}, we have $\|Z(j)\|_p = O(r)$, and so $|\sigma(HH(j))
x_{HH(j)}-y_j|^p = O(r^p)$, which as argued above, is small enough for
$N = \poly(k/\eps)$ sufficiently large.
\\\\
The proof of our theorem follows by a union bound over the events that
we defined.
\end{proof}

\ifconf
\else
\section{Adaptively Finding a Duplicate in a Data Stream}
We consider the following {\sf FindDuplicate} problem. We are given an adversarially ordered stream $\mathcal{S}$ of $n$ elements in $\{1, 2, \ldots, n-1\}$ and the goal is to output an element that occurs at least twice, with probability at least $1-\delta$. We seek to minimize the space complexity of such an algorithm. We improve the space complexity of \cite{jst11} for {\sf FindDuplicate} from $O(\log^2 n)$ bits to $O(\log n)$ bits, though we use $O(\log \log n)$ passes instead of a single pass. Notice that \cite{jst11} also proves a lower bound of $\Omega(\log^2 n)$ bits for a single pass.  

We use Lemma~\ref{lemma:onesparse} of our multi-pass sparse recovery
algorithm:
\begin{fact}\label{fact:cs}
Suppose there exists an $i$ with $|x_i| \geq C \|x_{[n] \setminus \{i\}}\|_2$ for some constant $C$. Then $O(\log \log n)$ adaptive measurements suffice to recover a set $T$ of constant size so that $i \in T$ with probability at least $1/2$. Further, all adaptive measurements are linear combinations with integer coefficients of magnitude bounded by $\poly(n)$.
\end{fact}

Our algorithm {\sf DuplicateFinder} for this problem considers the equivalent formulation of {\sf FindDuplicate} in which we think of an underlying frequency vector $x \in \{-1, 0, 1, \ldots, n-1\}^n$. We start by initializing $x_i = -1$ for all $i$. Each time item $i$ occurs in the stream, we increment its frequency by $1$. The task is therefore to output an $i$ for which $x_i > 0$. 

\begin{theorem}
There is an $O(\log \log n)$-pass, $O(\log n \log 1/\delta)$ bits of space per pass algorithm for solving the {\sf FindDuplicate} problem with probability at least $1-\delta$. 
\end{theorem}
\begin{proof}
We describe an algorithm {\sf DuplicateFinder} which succeeds with probability at least $1/8$. Since it knows whether or not it succeeds, the probability can be amplified to $1-\delta$ by $O(\log 1/\delta)$ independent parallel repetitions. It is easy to see that the pass and space complexity are as claimed, so we prove correctness. 
%
%
\begin{center}
\fbox{
\parbox{6.5in} {
\underline{DuplicateFinder}($\mathcal{S}$)
\begin{enumerate}
\item Repeat the following procedure $C = O(1)$ times independently.
\begin{enumerate}
\item Select $O(1)$-wise independent uniform $t_i \in [0,1]$ for $i \in [n]$.
\item Let $\eps > 0$ be a sufficiently small constant. Let $m = O(\log 1/\eps)$. 
\item Let $z^1, \ldots, z^{4m}$ be a pairwise-independent partition of the coordinates of $z$, where $z = x_i/t_i$ for all $i$. 
\item Run algorithm $A$ independently on vectors $z^1, \ldots, z^{4m}$. 
\item Let $T_1, \ldots, T_{4m}$ be the outputs of algorithm $A$ on $z^1, \ldots, z^{4m}$, respectively,
as per Fact \ref{fact:cs}. 
\item Compute each $x_i$ for $i \in \cup_{j=1}^{4m} T_j$ in an extra pass. If there is an $i$ for which $x_i > 0$, then output $i$.
\end{enumerate}
\item If no coordinate $i$ has been output, then output {\sc fail}.
\end{enumerate}
}}
\end{center}
We use the following fact shown in the proof of Lemma 3 of \cite{jst11}.
\begin{lemma}(see first paragraph of Lemma 3 of \cite{jst11})\label{lem:prev}
For a single index $i \in [n]$ and $t$ arbitrary, we have 
$$\Pr[\|z_{\overline{H_m(z)}}\|_2 > \frac{1}{20} \sqrt{m} \|x\|_1 \mid t_i = t] = O(\eps),$$
\end{lemma}
where $H_m(z)$ denotes the set of $m$ largest (in magnitude) coordinates of $z$. Suppose $|z_i| > \|x\|_1$ for some value of $i$. This happens if
$t_i < \frac{|x_i|}{\|x\|_1}$ and occurs with probability equal to $\frac{|x_i|}{\|x\|_1}$. 
Conditioned on this event, by Lemma \ref{lem:prev} we have that with probability $1-O(\eps)$,
$$\|z_{\overline{H_m(z)}}\|_2 \leq \frac{1}{20} \sqrt{m} \|x\|_1.$$
Suppose $i$ occurs in $z^j$ for some value of $j \in [4m]$. Since the partition is pairwise-independent,
the expected number of $\ell \in H_m(z) \setminus \{i\}$ which occur in $z^j$ is at most $\frac{m}{4m}$, 
and so with probability at least $3/4 - O(\eps)$, the norm of $z^j$ with coordinate corresponding to coordinate $i$ in $z$ removed is at most
$$\|z_{\overline{H_m(z)}}\|_2 \leq \frac{1}{20} \sqrt{m} \|x\|_1 \leq \frac{1}{20} \sqrt{m} |z_i|.$$ 
Since $m$ is a constant, by Fact \ref{fact:cs}, with probability at least $1/2$, $A$ outputs a set $T$ which 
contains coordinate $i$. Hence, with probability at least $3/8 - O(\eps)$, if there is an $i$ for which 
$|z_i| > \|x\|_1$, it is found by {\sf DuplicateFinder}. 

Let $p_i = \frac{|x_i|}{\|x\|_1}$. Then $|z_i| > \|x\|_1$ with probability $p_i$.
Since $\sum_i x_i > 0$, we have $\sum_{i \ \mid \ x_i > 0} p_i > \frac{1}{2}$. Consider one of the 
$C = O(1)$ independent repetitions of step 1.  

For coordinates $i$ for which $x_i > 0$, let $W_i = 1$ if $|z_i| > \|x\|_1$, and let $W = \sum_i W_i$. 
Then ${\bf E}[W] > 1/2$ and by pairwise-independence, ${\bf Var}[W] \leq {\bf E}[W]$. 

Let $W'$ be the average of the random variable $W$ over $C$ independent repetitions. 
Then ${\bf E}[W'] = {\bf E}[W] > \frac{1}{2}$ and 
${\bf Var}[W'] \leq \frac{{\bf E}[W]}{C}$, and so by Chebyshev's inequality for $C = O(1)$ sufficiently
large we have that with probability at least $\frac{1}{2}$, $W' > 0$, which means that in one of the $C$ repetitions 
there is a coordinate $i$ for which $x_i > 0$ and $|z_i| > \|x\|_1$. 

Hence, the overall 
probability of success is at least $1/2 \cdot (3/8-O(\eps)) > 1/8$, for $\eps$ sufficiently small. This 
completes the proof. 
\end{proof}
\fi

\section*{Acknowledgements}
This material is based upon work supported by the Space and Naval
Warfare Systems Center Pacific under Contract No. N66001-11-C-4092,
David and Lucille Packard Fellowship, MADALGO (Center for Massive Data
Algorithmics, funded by the Danish National Research Association) and
NSF grant CCF-1012042.  E. Price is supported in part by an NSF
Graduate Research Fellowship.

\ifconf
\else
\bibliographystyle{alpha}
\fi
\bibliography{adaptive}

\end{document}